\documentclass[reqno]{amsart}
\usepackage{graphicx}

\newtheorem{lemma}{Lemma}

\newtheorem{hypothesis}{Hypothesis}
\newtheorem{remark}{Remark}

\newcommand{\be}{\begin{eqnarray}}
\newcommand{\ee}{\end{eqnarray}}
\newcommand{\bee}{\begin{eqnarray*}}
\newcommand{\eee}{\end{eqnarray*}}
\newcommand{\R}{{\mathbb R}}

\newcommand{\Z}{{\mathbb Z}}

\newcommand{\x}{\mbox {\rm x}}
\newcommand{\y}{\mbox {\rm y}}
\newcommand{\I}{\mbox {\sc 1}}

\newcommand{\asy}{{\mathcal O}}
\newcommand{\da}{d_A}

\begin{document}

\title [] {Stationary solutions for multi-dimensional Gross-Pitaevskii equation with double-well potential}

\author {Andrea SACCHETTI}

\address {Department of Physics, Computer Sciences and Mathematics, University of Modena e Reggio Emilia, Modena, Italy.}

\email {andrea.sacchetti@unimore.it}

\date {\today}

\begin {abstract} In this paper we consider a non-linear Schr\"odinger equation with a cubic nonlinearity and a multi-dimensional double well potential. \ In the semiclassical limit the problem of the existence of stationary solutions simply reduces to the analysis of a finite dimensional Hamiltonian system which exhibits different behavior depending on the dimension. \ In particular, in dimension 1 the symmetric stationary solution shows a standard pitchfork bifurcation effect, while in dimension 2 and 3 new asymmetrical solutions associated to saddle points occur. \ These last solutions are localized on a single well and this fact is related to the phase transition effect observed in Bose-Einstein condensates in periodical lattices.

\bigskip

{\it Ams classification (MSC 2010):} 35Q55; 81Q20 

\bigskip


\end {abstract}

\maketitle

\section {Introduction}

Atomic Bose-Einstein condensates (BECs) are described by means of nonlinear Schr\"odinger equations where the nonlinear term of the form $ |\psi|^{2\sigma} \psi $, 
$\sigma=1,2,\ldots $, represents the $(\sigma+1)$-body contact potential \cite {Kohler}, where $\psi$ is the condensate's wavefunction. \ In fact, BECs strongly depend by interatomic forces and the binary coupling term $ |\psi|^2 \psi$ usually represents the dominant nonlinear term, the nonlinear Schr\"odinger equation obtained for $\sigma=1$ takes the usual form of the well-known Gross-Pitaevskii equation \cite {PitStr}.  

The analysis of the time-dynamic of BECs is, in general, an open problem and few rigorous results has been given. \ Among the basic models for BECs the model with a symmetric external \emph {double-well} potential plays an important role. \ Indeed, the explanation of some basic properties in such a relatively simple model will enable us to understand the fundamental mechanisms for a large family of BECs. \ For instance, the \emph {phase transition phenomenon} we can observe for BECs in a periodic lattice can be explained as a result of the bifurcation effects we can already see in the relatively simple double well model. \ In particular, for BECs in a periodic lattice has been seen a transition from the \emph {superfluidity phase} to the \emph {Mott-insulator phase} when the effective nonlinearity parameter becomes larger than a critical value \cite {B1,F,G}. \ In particular, it turns out that such a transition is quite slow in one-dimensional lattice, while it becomes very sharp in three-dimensional lattices. 

We consider here the case where the external potential of the linear part of the Schro\"odinger equation has a \emph {double well} shape. \ 
If the nonlinear term is absent then the linear Schr\"odinger equation has symmetric and antysimmetric eigenstates. \ However, the introduction of a nonlinear term,  which usually models in quantum mechanics an interacting many-particle system, may give rise to asymmetrical states related to spontaneous symmetry breaking phenomenon. \ It has been already proved that for one-dimensional nonlinear Schr\"odinger equations with double-well potentials  (see \cite {FS,Sacchetti2} for the result obtained in the semiclassical limit, see also \cite {KKSW} in the limit of large barrier between the two wells) then the symmetric/antisymmetric stable stationary state bifurcates when the adimensional effective nonlinear parameter takes absolute value equal to a critical value.

In this paper we explore in detail the different pictures may occur in dimension 1, 2 and 3 for the stationary solutions of the Gross-Pitaevskii equation associated to the linear ground state. \ In dimension 1 the only situation can occur is the bifurcation of the symmetric stationary solution, where a branch of asymmetrical solutions appears and where these asymmetrical solutions are going to be gradually localized on a single well when the nonlinearity strength increases. \ In dimension 2 we have different kind of bifurcations.  \ One kind of bifurcations is similar to the one already seen in dimension 1. \ Furthermore, new families of bifurcations appear, they are associated to the spontaneous symmetry breaking effect where a saddle point appears for a critical value of the nonlinearity strength. \ The stationary solutions on the branches raising from such a saddle point have the following peculiarity: they are mostly localized on just one well and thus the localization effect suddenly occurs when the nonlinearity parameter is around its critical value. \ A similar picture, with a more intricate sequence of bifurcations, occurs in dimension 3, too. 

The different behavior between models in dimension 1 and in dimensions higher than 1 has an important physical consequences when we consider BECs in lattices. \ Indeed has been oberved that for BECs in lattices a transitions from a superfluidity phase to a Mott-insulator phase occurs when the nonlinearity strength reaches a critical values; in particular in dimension 1 the transition is smooth, while in dimension 3 the transition is sharp. \ In fact, such a different behavior is expected to be connected to the appearance, in dimension 2 and 3, of stationary solutions localized on a single lattice site as we have seen for double-well models.

The paper is organized as follows. \ In Section \ref {Sec1} we introduce the model. \ In Section \ref {Sec2} we consider the $N$-mode approximation for nonlinear Schr\"odinger operator with a lattice potential in any dimension $d$. \ In Section \ref {Sec3} we consider in more detail the $N$-mode approximation for double-well potential in any dimension. \ Finally, in Section \ref {Sec4} we numerically compute the stationary solutions of the $N$-mode approximation in dimension $d=1$, $d=2$ and $d=3$ associated to the linear ground state.

\section {Double-well model} \label {Sec1}

Here, we consider the nonlinear Schr\"odinger (hereafter NLS) equation in the $d$-dimensional space $\x =(x_1, \ldots , x_d) \in \R^d$
\be
\left \{ 
\begin {array}{l}
i \hbar \frac {\partial \psi }{\partial t} = H_0\psi  + \epsilon |\psi |^{2 \sigma} \psi \, ,  \ \psi (\cdot , t) \in L^2 (\mathbb{R}^d , d\x ), \ \| \psi (\cdot , t ) \| =1 \\ 
\psi_0 (\x ) = \psi (\x , 0) 
\end {array}
\right.  ,  
\label {F1}
\ee
where $\epsilon \in \mathbb{R}$ and $\| \cdot \| $ denotes the $L^2 (\mathbb{R}^d , d\x )$ norm; 
\
\be
H_0 = \frac {-\hbar^2}{2m} \Delta + V ,\ \ \Delta = \sum_{j=1}^d \frac {\partial^2}{\partial x_j^2} \label {F2}
\ee
is the linear Hamiltonian with a \emph {lattice potential} $V(x)$. \ In the case of cubic nonlinearity where $\sigma =1$ then (\ref {F1}) is usually called  Gross-Pitaevskii equation. \ For the sake of definiteness we assume the units such that $2m=1$. \ The semiclassical parameter $\hbar >0$ is such that $\hbar \ll 1$.

Let us introduce the assumptions on the lattice potential $V(x)$.

\begin {hypothesis} \label {Hyp1} Let $v(\x ) \in C_0^\infty (\R^d )$ be a spherically symmetric single well potential, that is $v(\x ) = f(|\x |)$ where $f (r)  \in C_0^\infty (\R^+ )$ is a smooth non-positive monotone not-decreasing function with compact support and such that $f(0)<0$. \ In particular we assume that $f'(0+0)=0$ and $f''(0+0)>0$. \ Then $v(\x ) $ is a smooth function with compact support and with a non-degenerate minimum value at $\x =0$:
\be
v (\x ) > v_{min} = f (0) , \ \forall \x \in \R^d  , \ \x \not= 0 \,. \label {F3}
\ee
By construction the support of $v(\x)$ is a $d$-dimensional ball with center at $\x =0$ and radius $a$, for some $a>0$. \ Let $J_m \in \Z$ and $ K_m \in \Z$, $m=1,2,\ldots , d$ be fixed and such that $J_m \le K_m$; let 
\bee
J = \left \{ j =(j_1 , \ldots , j_d ) \in \Z^d \ : \ J_m \le j_m \le K_m , \ m=1,2,\ldots , d \right \} .
\eee
We then define a \emph {lattice potential} as
\be
V(\x )= \sum_{j\in J} v(\x - \x_j) , \ \x_j = jb =(j_1 b, \ldots , j_d b) \, , \label {F4}
\ee
where $b>0$ is such that $b >2a$. \ Hence, by construction, the lattice potential $V(\x)$ has exactly 
\be
N:= \Pi_{m =1}^d [K_m - J_m +1] \label {F4Bis}
\ee
similar wells with non-degenerate minima at $\x = \x_j$, $j\in J$.
\end {hypothesis}

It is a well known fact (\cite {CW} and Thm. 6.2.1 \cite {C}) that the Cauchy problem (\ref {F1}) is globally well-posed for any $\epsilon \in \R$ small enough provided that
\bee
\sigma < \left \{ 
\begin {array}{ll}
+\infty & \ \mbox { if } \ d \le 2 \\ 
\frac {2}{d-2} & \ \mbox { if } \ d > 2 
\end {array}
\right. . 
\eee
In such a case the conservation of the norm of $\psi (\x , t)$ and of the energy
\bee
{\mathcal E} (\psi ) = \langle \psi , H_0 \psi \rangle + \frac {\epsilon}{\sigma +1} \langle \psi^{\sigma+1} ,  \psi^{\sigma +1} \rangle 
\eee
follows.

\section {Reduction to the $N$-mode approximation} \label {Sec2}

Now, making use of the semiclassical analysis \cite {H} we reduce the NLS equation (\ref {F1}) to a $N$-dimensional Hamiltonian system, usually denoted $N$-mode approximation, where $N$ is the total number of lattice sites defined in (\ref {F4Bis}). \ We make use of the ideas already developed in the papers \cite {BS,S} and adapted here to the case of a lattice potential (\ref {F4}). \ Since the reduction method is similar to the one already exploited in \cite {BS,S} then we don't dwell here on the details of the proof of the validity of the reduction to the $N$-mode approximation and simply we state the main results.

\subsection {Semiclassical results}

One of the mail tools in semiclassical analysis is the notion of Agmon distance. \ Let $E\ge v_{min}$ be fixed, then  the Agmon (pseudo-)distance, associated to the energy $E$, between two points $\x$ and $\y$ is defined as 
\bee
\da (\x ,\y ; E) = \inf_\gamma \int_\gamma \sqrt {\left [ V(x) -E \right ]_+} dx 
\eee
where the $\inf$ is taken on the set of all regular paths $\gamma$ connecting the two points $\x$ and $\y$ and where $[V-E]_+ = \frac {|V-E|+V-E}{2}$. \ Let us consider the Agmon distance associated to the ground state energy, since the difference  between the ground state energy and the minimum of the potential is of order $O(\hbar )$ then in the semiclassical limit we can choose $E=v_{min}$; in the following, for the sake of simplicity, let us denote
\bee
\da (\x , \y ):=\da (\x , \y ; v_{min}) \, . 
\eee
We then define the following two quantities
\bee
S_0 = \inf_{j\not= \ell \, , \ j,\ell \in J} \da (\x_j ,\x_\ell )
\eee
and
\bee
S_1 = \inf_{|j - \ell |>1 \, , \ j,\ell \in J} \da (\x_j ,\x_\ell ), \ \mbox { where } \ |j - \ell | = \sum_{m=1}^d |j_m - \ell_m| ;
\eee
then, by construction of the lattice potential $V(\x )$, the following result holds true.

\begin {lemma} \label {Lemma1}
Let $j\in J$ and let $\ell \in J$ be such that $|j-\ell |=1$, then 
\be
S_0 = \da (\x_j , \x_\ell ) = 2 \int_0^{b/2} \sqrt {f(r) - f(0)} dr \label {F5}
\ee
is independent of $j$ and $\ell$, and 
\be
S_0 < S_1 \, . \label {F6}
\ee
\end {lemma}

\begin {proof}
If $|j-\ell |=1$ then all the components of $\x_j$ and $\x_\ell$ are equal, but one: e.g. $j_m =\ell_m$ for $m=2,\ldots ,d$ and $j_1-\ell_1 =  1$. \ Then, by construction it turns out that 
\bee
\da (\x_j , \x_\ell ) &=& \int_0^b \sqrt {V[\x_j +r (1,0,\ldots ,0)]-v_{min}} dr \\ 
&=& 2 \int_0^{b/2} \sqrt {v[r (1,0,\ldots ,0)]-v_{min}} dr \\ 
&=& 2 \int_0^{b/2} \sqrt {f(r)-v_{min}} dr 
\eee
proving (\ref {F5}). \ Now, in order to prove (\ref {F6}) let 
\be
B_{\x_j} (S_0/2) = \left \{ \x \in \R^d \ : \ \da (\x_j ,\x ) \le \frac 12 S_0 \right \}\, ; \label {F6Bis}
\ee
then, by construction, it follows that 
\bee
B_{\x_j} (S_0/2) = \left \{ \x \in \R^d \ : \ |\x_j -\x | \le \frac 12 b \right \}
\eee
where $|\x-\x_j|$ denotes here the usual distance in $\R^d$ between two points $\x$ and $\x_j$. \ Hence, if $|j-\ell |>1$ then 
\bee
B_{\x_j}(S_0/2) \cap B_{\x_\ell }(S_0/2) = \emptyset 
\eee
and thus (\ref {F6}) follows.
\end {proof}

\begin {lemma} \label {Lemma2}
Let $j =(j_1,\ldots , j_d)\in \Z^d$, let 
\bee
\lceil j \rceil := \max_{m=1, \ldots ,d} |j_m|\, .
\eee
Then 
\be
\da (\x_j , \x_\ell ) \ge \lceil j-\ell \rceil S_0\, . \label {F7}
\ee
\end {lemma}

\begin {proof}
Assume that $s:= \lceil j-\ell \rceil \not= 0$, otherwise $j=\ell $ and (\ref {F7}) immediately follows. \ Assume also, for argument's sake, that $|j_1 - \ell_1 |=s$ and that $J_m \le 0 \le K_m$ for any $m$. \ By construction of the lattice potential it turns out that 
\bee
\sqrt {V(x_1,x_2, \ldots , x_d) -v_{min}} \ge \sqrt {V(x_1,0, \ldots , 0) -v_{min}} \, , 
\eee
hence, for any regular path $\gamma$ connecting the two points $\x_j$ and $\x_\ell$, it follows that
\bee
\da (\x_j , \x_\ell ) 
&\ge & \int_\gamma \sqrt {V(x_1,x_2, \ldots , x_d) -v_{min}} dx \\ 
&\ge & \int_\gamma \sqrt {V(x_1,0, \ldots , 0) -v_{min}} dx \\ 
&\ge & \int_0^s \sqrt {V(x_1,0, \ldots , 0) -v_{min}} dx_1 \\
&=& s\cdot S_0 \, . 
\eee
\end {proof}

\subsection {$N$-mode approximation}

Now, let $H_D$ be the Dirichlet realization of the Schr\"odinger operator formally defined on $L^2 (B_S(0),d\x )$ by
\be
H_D = - \hbar^2 \Delta + v \label {F8}
\ee
where $B_S(0)$ is the ball with center at $\x =0$ and radius $S >2 S_0$, as defined in (\ref {F6Bis}). \ Since the bottom of $v(\x)$ is not degenerate, then the Dirichlet problem associated to the single-well trapping potential $v(\x )$ has spectrum $\sigma (H_D)$ with ground state 
\bee
\lambda_D = v_{min} + d \sqrt {\mu} \hbar + \asy (\hbar^2) \, , \ \mu =\frac 12 f''(0),
\eee
such that 
\bee
\mbox {dist} \left [ \lambda_D , \sigma (H_D) \setminus \{ \lambda_D \} \right ] \ge 2 C \hbar 
\eee
for some $C>0$. \ The normalized eigenvector $\psi_D (\x )$ associated to $\lambda_D$ is localized in a neighborhood of $\x =0$ and it exponentially decreases as $\asy \left ( \hbar^{-m} e^{-\da (\x )/\hbar } \right )$ for some $m>0$, and where $\da (\x) := \da (\x , 0)$ is the Agmon distance between $\x$ and the point $\x =0$. \ In particular, in a neighborhood of $\x =0$ then $\psi_D (\x ) \sim \mu^{d/8} (\pi \hbar)^{-d/4} e^{-\sqrt {\mu}|x|^2/2\hbar}$.

The bottom of the spectrum $\sigma (H_0)$ of $H_0$ contains exactly $N$ eigenvalues $\lambda_j$, $j \in J$, such that 
\bee
\lambda_j - \lambda_D = \asy (e^{-\rho/\hbar}) 
\eee
for any $0<\rho < S_0$; this result is a consequence of the fact that the multiple well potential $V(\x)$ is given by a superposition of $N$ exactly equal wells displaced on a regular lattice. \ Furthermore
\bee
\mbox {dist} \left [ \{ \lambda_j \}_{j\in J}, \sigma (H_0) \setminus \{ \lambda_j \}_{j\in J} \right ] > C \hbar \, . 
\eee

Let $F$ be the eigenspace spanned by the eigenvectors $\psi_j$ associated to the $N$ eigenvalues $\lambda_j$. \ Then, the restriction $H_0|_F$ of $H_0$ to the subspace $F$ can be represented in the basis of orthonormalized vectors $\phi_j$, $j \in J$, such that
\be
\phi_j (\x ) - \varphi_j (\x ) = \asy (e^{-\rho /\hbar }) \ \mbox { where } \ \varphi_j (\x )= \psi_D (\x - \x_j ); \label {F9}
\ee
hence, the vector $\phi_j (\x )$ is localized in a neighborhood of the minima point $\x_j$. \ More precisely, in the semiclassical limit it follows that (for the proof we refer to Theorem 4.3.4 and Theorem 4.4.6 by \cite {H}).

\begin {lemma} \label {Lemma3}
Up to an error of order $\asy ( \hbar^\infty  e^{-S_0/\hbar} )$ the restriction $H_0|_F$ of $H_0$ to the subspace $F$ is represented in the basis $\phi_j (\x )$, $j \in J$, by the square matrix $T$ defined as
\be
T_{j,\ell } = 
\left \{
\begin {array}{ll}
\lambda_D & \mbox { if } |j-\ell | =0 \\ 
- \beta & \mbox { if } |j-\ell | = 1 \\
0 & \mbox { if } |j-\ell |> 1
\end {array}
\right. \, , \label {F10}
\ee
where $\beta$ is a quantity independent on the indexes and such that 
\be
\frac 1C \hbar^{1/2} \le \beta e^{S_0/\hbar} \le C \hbar^{1-d/2} \, . \label {F12}
\ee
\end {lemma}

Let $\psi$ be the normalized solution of the NLS equation (\ref {F1}). \ Then $\psi$ may be written in the following form. \ Let $\Pi$ be the projection operator on the space $F$, and let $\Pi_c = \I - \Pi$. \ If the initial state $\psi_0 (\x )=\psi (\x ,0)$ is  prepared on the space $F$ spanned by the $N$ ground state linear eigenvectors, that is 
\bee
\psi_c (\x ,0) =0, \ \mbox { where } \ \psi_c = \Pi_c \psi \, , 
\eee
then it is possible to prove, by making use of ideas similar to the ones developed by \cite {BS,S}, that $\psi_c (\x , t)$ is exponentially small for times of order $\beta^{-1}$, that is 
\bee
\| \psi_c (\cdot , t ) \|_{L^2 (\R^d)} = \asy (e^{-S_0/\hbar }) , \ \forall t \in [0,\beta^{-1}]\, , 
\eee
and that $\Pi \psi$ can be written in the form 
\bee
\Pi \psi (\x ,t) = \sum_{j\in J} d_j (t) \varphi_j (\x ) + \asy \left ( \hbar^\infty e^{-S_0/\hbar } \right ) 
\eee
where $d_j (t)$, $j \in J$, satisfy to the \emph{ $N$-mode approximation} for the NLS equation (\ref {F1}), which consists in to the following system of ODEs
\be
i \hbar \dot d_j = \sum_{\ell \in J}T_{j,\ell} d_\ell + \epsilon c |d_j|^{2\sigma} d_j ,\ j\in J , \label {F13}
\ee
with the normalization condition
\be
\sum_{j\in J} |d_j (t)|^2 =1\, ; \label {F14}
\ee
where %
\bee
c= \| \varphi_j \|_{L^{2\sigma +2}}^{2\sigma +2} = \| \psi_D \|_{L^{2\sigma +2}}^{2\sigma +2}
\eee
is a real valued constant independent of $j$. 

\subsection {Hamiltonian form of the $N$-mode approximation}

If we set 
\bee
d_j = \sqrt {q_j} e^{i\theta_j} ,\ q_j \in [0,1],\ \theta_j \in [0,2\pi ) \, , 
\eee
then, by means of a straightforward calculation, it follows that (\ref {F13}) takes the \emph {Hamiltonian form}
\be
\left \{
\begin {array}{lcl}
\hbar \dot q_j &=& \frac {\partial {\mathcal H}}{\partial \theta_j} = - 2 \beta \sum_{\ell \in J : |j-\ell |=1} \sqrt {q_j q_\ell} \sin (\theta_\ell - \theta_j ) \\ 
\hbar \dot \theta_j &=& - \frac {\partial {\mathcal H}}{\partial q_j} = -\lambda_D + \beta \sum_{\ell \in J : |j-\ell |=1} \sqrt {\frac {q_\ell}{q_j}} \cos (\theta_\ell - \theta_j ) - \epsilon c q_j^\sigma
\end {array}
\right.
\label {F15}
\ee
with Hamiltonian function
\be
{\mathcal H} &:=& {\mathcal H} (q,\theta ) \nonumber \\ 
&=& \lambda_D \sum_j q_j - \beta \sum_{j, \ell \in J : |j-\ell |=1} \sqrt {q_j q_\ell} \cos (\theta_\ell - \theta_j )+ \frac {c\epsilon}{\sigma+1} \sum_{j\in J} q_j^{\sigma+1} \, . \label {F16}
\ee
The normalization condition (\ref {F14}) takes the form
\be
\sum_j q_j =1 \, \label {F17}
\ee
and the Hamiltonian function is a constant of motion, i.e.
\bee
{\mathcal H} [q(t),\theta (t)] = {\mathcal H} [q(t_0),\theta (t_0)]
\eee
for any $t$, $t_0=0$ is the initial instant.

\subsection {Stationary solutions} 

Stationary solutions are the normalized solutions of (\ref {F1}) of the form $\psi (\x ,t)=e^{i\omega t} \psi (\x )$. \ Concerning the study of the stationary solutions has been proved by \cite {FS} that for $1$-dimensional double-well models then the 2-mode approximation gives the stationary solutions for the NLS (\ref {F1}), up to an exponentially small error; furthermore the orbital stability of the stationary solutions is proved. \ The same arguments may be applied to the general $d$-dimensional problem with lattice potential; that is the stationary solutions of $N$-mode approximation (\ref {F13}) and (\ref {F14}), for any $N\ge 2$, give, up to an exponentially small error $\asy (e^{-\rho/\hbar})$, for any $0<\rho < S_0$, the stationary solutions of the NLS (\ref {F1}).

In terms of $N-$mode approximation (\ref {F13}) it consists in looking for the solution of the system of equations 
\be
- \omega \hbar d_j = \sum_{\ell \in J} T_{j,\ell } d_\ell + \epsilon c |d_j|^{2\sigma} d_j , \ j \in J \, . \label {F18}
\ee
As before, if we set $d_j = \sqrt {q_j}e^{i\theta_j}$, then $q_j$ and $\theta_j$ must be the solution of the system of equations
\be
\left \{
\begin {array}{lcl} 
0 &=& \frac {\partial {\mathcal H}}{\partial \theta_j} \\ 
- \hbar \omega &=& - \frac {\partial {\mathcal H}}{\partial q_j}
\end {array}
\right. \, . 
\label {F19}
\ee
Finally, if we set
\bee
\Omega = \frac {\hbar \omega - \lambda_D}{\beta} \ \mbox { and } \ \eta = \frac {c \epsilon}{\beta}
\eee
then finally we get the equations for stationary solutions
\be
\left \{
\begin {array}{lcl}
0 &=&  \sum_{\ell \in J : |j-\ell |=1} \sqrt {q_j q_\ell} \sin (\theta_\ell - \theta_j ) \\ 
-\Omega &=&  \sum_{\ell \in J : |j-\ell |=1} \sqrt {\frac {q_\ell}{q_j}} \cos (\theta_\ell - \theta_j ) - \eta  q_j^\sigma
\end {array}
\right. \label {F19Bis}
\ee

\section {Multi-dimensional double-well potential} \label {Sec3}

We consider now the basic model of multi-dimensional double-well potentials, where the lattice potential has exactly $2^d$ wells, that is we assume that $L_m=0$ and $K_m=1$ for any $m=1,\ldots , d$. \ In this case the matrix $T$ has a special form and its eigenvalues can be explicitly computed.

\subsection {One-dimensional model} In such a case the potential $V(\x)$ is a simply double-well potential with minima points $\x_0$ and $\x_1$ and the matrix $T:=T_1$ simply reduces to 
\bee
T_1 = \left ( 
\begin {array}{cc} 
\lambda_D & -\beta \\ -\beta & \lambda_D 
\end {array} 
\right )
\eee
The matrix $T_1$ has eigenvalues $\mu_1 = \lambda_D -\beta$ and $\mu_2 = \lambda_D + \beta$ with associated normalized eigenvectors $v_1 =\left ( \frac {1}{\sqrt {2}}, \frac {1}{\sqrt {2}} \right )$ and $v_1 =\left ( \frac {1}{\sqrt {2}}, -\frac {1}{\sqrt {2}} \right )$.

\subsection {Two-dimensional model} In such a case the potential $V(\x)$ has $4$ wells with minima points
\bee
\x_{(0,0)},\ \x_{(0,1)},\ \x_{(1,0)},\ \x_{(1,1)}
\eee
and the matrix $T:=T_2$ reduces to
\bee
T_2 = 
\left (
\begin {array}{cccc}
\lambda_D & - \beta & -\beta & 0  \\ 
-\beta & \lambda_D & 0 & -\beta \\ 
-\beta & 0 & \lambda_D & -\beta \\ 
0 & - \beta & - \beta & \lambda_D 
\end {array} 
\right )
=
\left ( 
\begin {array}{cc}
T_1 & - \beta \I_2 \\ -\beta \I_2 & T_1 
\end {array}
\right )
\eee
where $\I_2$ is the $2\times 2$ identity matrix. \ The matrix $T_2$ has eigenvalues $\mu_r$ and associated normalized eigenvectors $v_r$, $r=1,\ldots ,4$, given by 
\be
\begin {array}{lcllcl}
\mu_1 &=& \lambda_D - 2 \beta & v_1 &=& \left ( \frac 12 , \frac 12 , \frac 12 , \frac 12 \right ) \\ 
\mu_2 &=& \lambda_D  & v_2 &=& \left ( 0 , -\frac 1{\sqrt 2} , \frac 1{\sqrt 2} , 0 \right ) \\ 
\mu_3 &=& \lambda_D  & v_3 &=& \left (  -\frac 1{\sqrt 2} , 0 , 0 , \frac 1{\sqrt 2}  \right ) \\
\mu_4 &=& \lambda_D + 2 \beta & v_4 &=& \left ( \frac 12 , -\frac 12 , -\frac 12 , \frac 12 \right ) 
\end {array} \label {F21Bis}
\ee

\subsection {Three-dimensional model} In the case of $d=3$ then the potential $V(\x)$ has $8$ wells with minima points 
\bee
\x_{(0,0,0)},\ \x_{(0,0,1)},\ \x_{(0,1,0)},\ \x_{(0,1,1)},\ \x_{(1,0,0)},\ \x_{(1,0,1)},\ \x_{(1,1,0)},\ \x_{(1,1,1)}\, , 
\eee
and the matrix $T:=T_3$ reduces to 
\bee
T_3 = \left ( 
\begin {array}{cc}
T_2 & - \beta \I_4 \\ -\beta \I_4 & T_2 
\end {array}
\right )
\eee
where $\I_4$ is the $4\times 4$ identity matrix. \ The matrix $T_3$ has eigenvalues $\mu_r$ and associated normalized eigenvectors $v_r$, $r=1,\ldots ,8$, given by
\be
\begin {array}{lcllcl}
\mu_1 &=& \lambda_D - 3 \beta & v_1 &=& \left ( \frac 1{\sqrt {8}} , \frac 1{\sqrt {8}} , \frac 1{\sqrt {8}} , \frac 1{\sqrt {8}} , \frac 1{\sqrt {8}} , \frac 1{\sqrt {8}} , \frac 1{\sqrt {8}} , \frac 1{\sqrt {8}} \right ) \\
\mu_2 &=& \lambda_D -  \beta & v_2 &=& \left ( 0 , 0 , - \frac 12 , - \frac 12 , \frac 12 , \frac 12 , 0 , 0 \right ) \\
\mu_3 &=& \lambda_D -  \beta & v_3 &=& \left ( 0 , - \frac 12 , 0 , - \frac 12 , \frac 12 , 0 , \frac 12 , 0 \right ) \\
\mu_4 &=& \lambda_D -  \beta & v_4 &=& \left ( - \frac 12 , 0 , 0 ,  \frac 12 , -\frac 12 , 0 , 0 , \frac 12 \right ) \\
\mu_5 &=& \lambda_D +  \beta & v_5 &=& \left ( 0 , \frac 12 , 0 ,  - \frac 12 , -\frac 12 , 0 , \frac 12 , 0 \right ) \\
\mu_6 &=& \lambda_D +  \beta & v_6 &=& \left (  \frac 12 , 0 , 0 ,  -\frac 12 , -\frac 12 , 0 , 0 , \frac 12 \right ) \\
\mu_7 &=& \lambda_D +  \beta & v_7 &=& \left ( 0 , 0 ,  \frac 12 , - \frac 12 , - \frac 12 , \frac 12 , 0 , 0 \right ) \\
\mu_8 &=& \lambda_D + 3 \beta & v_8 &=& \left ( -\frac 1{\sqrt {8}} , \frac 1{\sqrt {8}} , \frac 1{\sqrt {8}} , -\frac 1{\sqrt {8}} , \frac 1{\sqrt {8}} , -\frac 1{\sqrt {8}} , -\frac 1{\sqrt {8}} , \frac 1{\sqrt {8}} \right ) 
\end {array} \label {F23Bis}
\ee

\subsection {Any dimension} In dimension $d+1$, for $d\ge 1$, the potential $V(\x)$ has $2^{d+1}$ wells with minima points 
\bee
\x_{(0,j')}, \ \x_{(1,j')}
\eee
where $j'\in \{0,1\}^d$ are the indexes of the model in dimension $d$. \ Then the matrix $T_{d+1}$ has the following form 
\be
T_{d+1} = \left ( 
\begin {array}{cc}
T_d & - \beta \I_{2^d} \\ -\beta \I_{2^d} & T_d 
\end {array}
\right ) 
\label {F20}
\ee
and it has $2^{d+1}$ eigenvalues (counting multiplicity). \ We state now a general result.

\begin {lemma} \label {Lemma4} Let $\Sigma^0 =\{ \lambda_D \}$ and, by induction, let
\bee
\Sigma^{d+1} = \left \{ \lambda \in \R \ : \ \exists \mu \in \Sigma^d \ \mbox { such that } \ |\lambda - \mu | = \beta \right \} \, ; 
\eee
that is:
\bee
\Sigma_{2d} &=& \left \{ \lambda_D , \lambda_D \pm 2\beta , \ldots , \lambda_D \pm 2 ( d -1) \beta , \lambda_D \pm 2 d \beta \right \} \\ 
\Sigma_{2d+1} &=& \left \{ \lambda_D \pm \beta, \lambda_D \pm 3\beta , \ldots , \lambda_D \pm (2  d -1) \beta , \lambda_D \pm (2 d+1) \beta \right \}
\eee
Then, the set of eigenvalues of $T_{d+1}$ coincides with $\Sigma^{d+1}$. \ Furthermore, if $\mbox {\rm mult} (\lambda)$ denotes the multiplicity of the eigenvalue $\lambda$, then
\bee
\mbox {\rm mult} (\lambda) = \sum_{\mu \in \Sigma^d \ : \ | \lambda - \mu | =\beta } \mbox {\rm mult} (\mu )\, . 
\eee
\end {lemma}

\begin {proof}
From (\ref {F20}) we have that the eigenvalue equation for $T_{d+1}$ is given by 
\be
0 = \left | T_{d+1} - \lambda \I_{2^{d+1}} \right | = 
\left | 
\begin {array}{cc} 
T_d -\lambda \I_{2^{d}}& -\beta \I_{2^{d}} \\ -\beta \I_{2^{d}} & T_d -\lambda \I_{2^{d}}
\end {array}
\right | \, . \label {F20Bis}
\ee
If we assume for a moment that $\lambda$ is not an eigenvalue of $T_d$ then from Schur's formula it follows that equation (\ref {F20Bis}) becomes 
\bee
0 = \left | T_{d} - (\lambda + \beta) \I_{2^d} \right | \cdot \left | T_{d} - (\lambda - \beta) \I_{2^d} \right |
\eee
from which follows that $\lambda \in \Sigma^{d+1}$ if, and only if, $\lambda \pm \beta \in \Sigma^d$. \ Since $\Sigma^0 = \{ \lambda_D \}$ has cardinality $1$, then $\Sigma^d$ has cardinality (counting multiplicity) $2^d$ and, by induction, we have that if $\lambda \in \Sigma^d$, then $\lambda \notin \Sigma^{d+1}$.
\end {proof}

\begin {remark}
It is not hard to see that the ground state associated to the eigenvalues $\lambda = \lambda_D - d \beta$ of $T_d$ has normalized eigenvector $v = \left ( 2^{-d/2} , \ldots , 2^{-d/2} \right )$.
\end {remark}

\section {Analysis of the bifurcation of the ground state} \label {Sec4}

Now, we are going to discuss in dimension 1, 2 and 3 how the ground state stationary solutions of the linear problem bifurcate when we introduce the nonlinear term. 

\subsection {One-dimensional model} The model in dimension $1$ has been largely discussed in previous papers (see, e.g., \cite {FS}), thus let us omit the details. \ In dimension $d=1$ then (\ref {F19Bis}) takes the form
\be
\left \{
\begin {array}{lcl}
\sqrt {q_1 q_2} \sin (\theta_2 - \theta_1 ) &=& 0 \\ 
- \sqrt {q_2/q_1} \cos (\theta_2 - \theta_1 ) + \eta q_1^\sigma &=& \Omega \\ 
- \sqrt {q_1/q_2} \cos (\theta_1 - \theta_2 ) + \eta q_2^\sigma &=& \Omega \\ 
q_1+q_2 &=& 1
\end {array}
\right. \label {F24Bis}
\ee
and for $\eta =0$ the above equation has ground state corresponding to 
\be
q_1 = q_2 = \frac 12 \, , \ \theta_1 = \theta_2 \ \mbox { and } \ \Omega = -1 +\frac {1}{2^\sigma}\eta \, . \label {F23x}
\ee
Now, we observe that all the solutions of (\ref {F24Bis}) are such that $q_1 \not= 0$ and $q_2 \not= 0$. \ Indeed, if there exists a solution of (\ref {F24Bis}) such that, for instance, $q_1 = 0$ and $q_2=1$ (or $q_1 = 1$ and $q_2=0$) at some $\eta $ then $d_1=0$ and $d_2 \not= 0$ in contradiction with (\ref {F18}) for $d=1$. \ Then stationary solutions are such that $\theta_1 = \theta_2$ or $\theta_1 - \theta_2 = \pi$ and the equation that gives the stationary solutions that bifurcate from the linear ground state simply reduces to
\be
\left \{
\begin {array}{lcl}
- \sqrt {q_2/q_1} + \eta q_1^\sigma &=& \Omega \\ 
- \sqrt {q_1/q_2} + \eta q_2^\sigma &=& \Omega \\ 
q_1+q_2 &=& 1
\end {array}
\right. \label {F23xx}
\ee
First of all we remark that the problem is invariant under the reflection
\bee
\x_0 \leftrightarrow \x_1 \, , \ \mbox { i.e. } \ q_1 \leftrightarrow q_2 \, . 
\eee
Hence, asymmetrical solutions, if there exists, are doubly degenerate.

If we set $z=q_1-q_2 \in (-1,+1)$, that is 
\bee
q_1 = \frac 12 + \frac 12 z \ \mbox { and } \ q_2 = \frac 12 - \frac 12 z
\eee
and if we set
\bee
\chi = \frac {\eta}{2^\sigma}
\eee
then equation (\ref {F23xx}) reduces to the form 
\be
\sqrt {\frac {1-z}{1+z}} - \chi \left ( 1 + z \right )^\sigma - \sqrt {\frac {1+z}{1-z}} + \chi \left ( 1 - z \right )^\sigma =0 \label {F26Bis}
\ee
and its solutions are given by 

\begin {itemize} 

\item [-] $z=0$, which coincides with the symmetric solution (\ref {F23x});

\item [-] for $\sigma =1$, $z= \pm \frac {1}{\chi} \sqrt {\chi^2-1}$, provided that $\chi <-1$. 

\end {itemize} 

Therefore,  for $\sigma=1$ we have that the symmetric ground state solution bifurcates at $\eta =-2$ and the new asymmetrical solutions are such that (see Fig. \ref {Fig1})
\be
z = \pm \frac {1}{\eta} \sqrt {\eta^2 - 4} \ \mbox { and } \ \Omega = \eta \, , \ \mbox { for } \eta \le -2 \, . 
\ee
In conclusion, for $\eta \ge -2$ the stationary solutions of equation (\ref {F19Bis}) for the cubic (corresponding to $\sigma =1$) one-dimensional double-well model are only given by the symmetric stationary solution (\ref {F23x}). \ At $\eta =-2$ this solution bifurcates and the new solutions are asymmetrical. \ The transition from the symmetric stationary solution to the asymmetrical stationary solution is smooth. \ In particular, in Table \ref {tabella1} we collect the values for $q_1$ and $q_2$ corresponding to the asymmetric stationary solution; it turns out that the asymmetrical stationary solutions become \emph {gradually} localized on only one of the two wells when $|\eta |>2$ increases (see also Fig. \ref {Fig1Bis}).
 
\begin{table}
\begin{center}
\begin{tabular}{|l||c|c|c|} \hline
$\eta $ 			& $q_1$ & $q_2$ & $\Omega$ \\ \hline  
$-2.01 $ & $0.45$ & $0.55$ & $   -2.01   $ \\ \hline 
$-2.1 $ & $0.35$ & $0.65$ & $   -2.1   $ \\ \hline 
$-2.5 $ & $0.2$ & $0.8$ & $   -2.5   $ \\ \hline
\end{tabular}
\caption{Values of the coefficients $q_1$ and $q_2$ of the asymmetrical stationary solutions for the cubic one-dimensional double-well model for some values of the parameter $\eta$.}
\label{tabella1}
\end{center}
\end {table}

\begin{figure}
\begin{center}
\includegraphics[height=6cm,width=6cm]{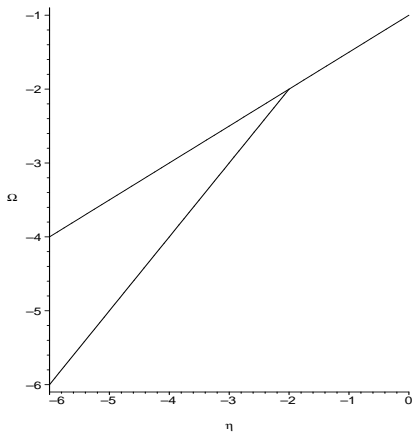}
\caption{Bifurcation picture of the energy $\Omega$ as function of $\eta$ for the cubic one-dimensional double-well  model.}
\label {Fig1}
\end{center}
\end{figure}

\begin{figure}
\begin{center}
\includegraphics[height=4cm,width=4cm]{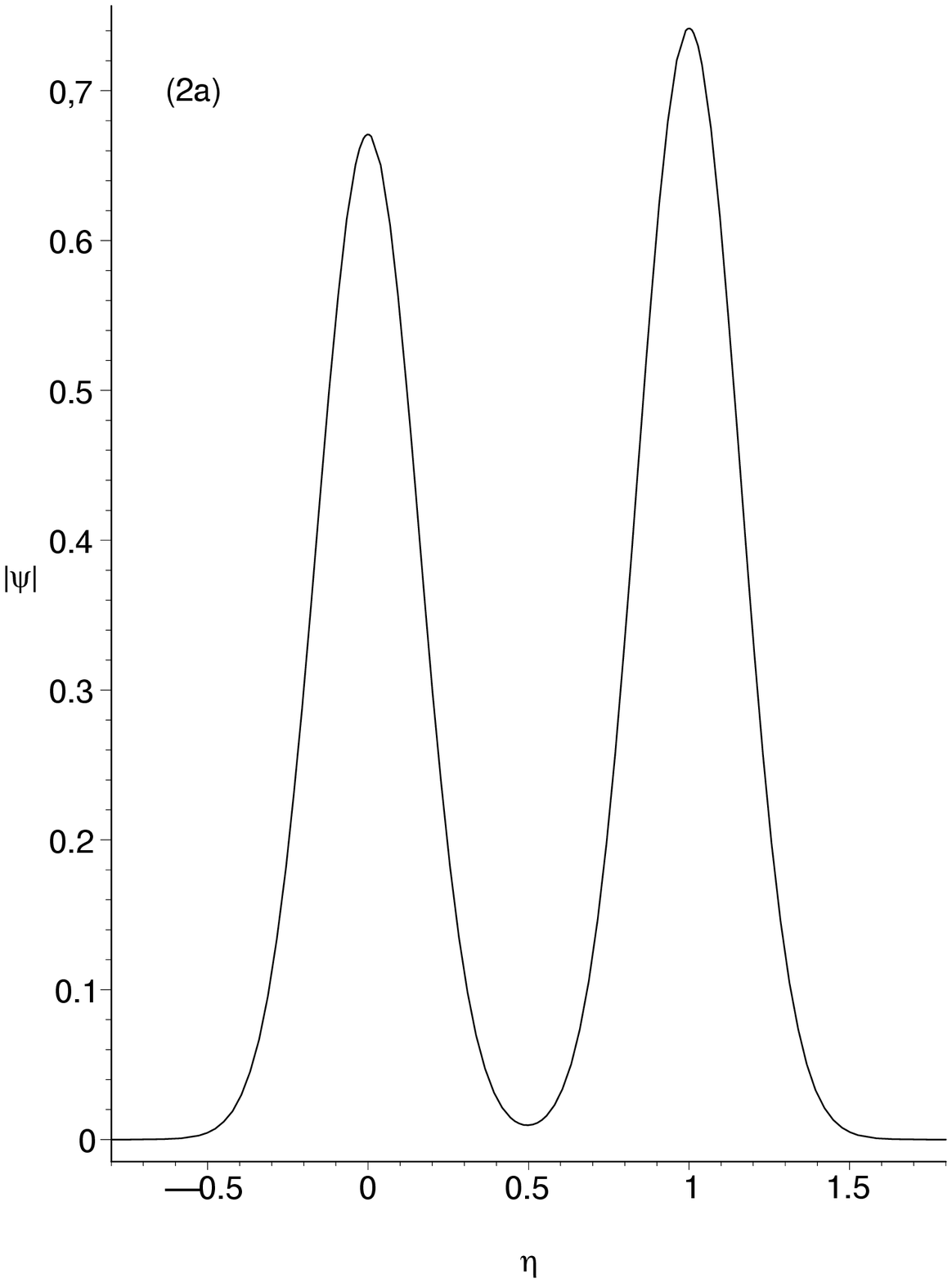}
\includegraphics[height=4cm,width=4cm]{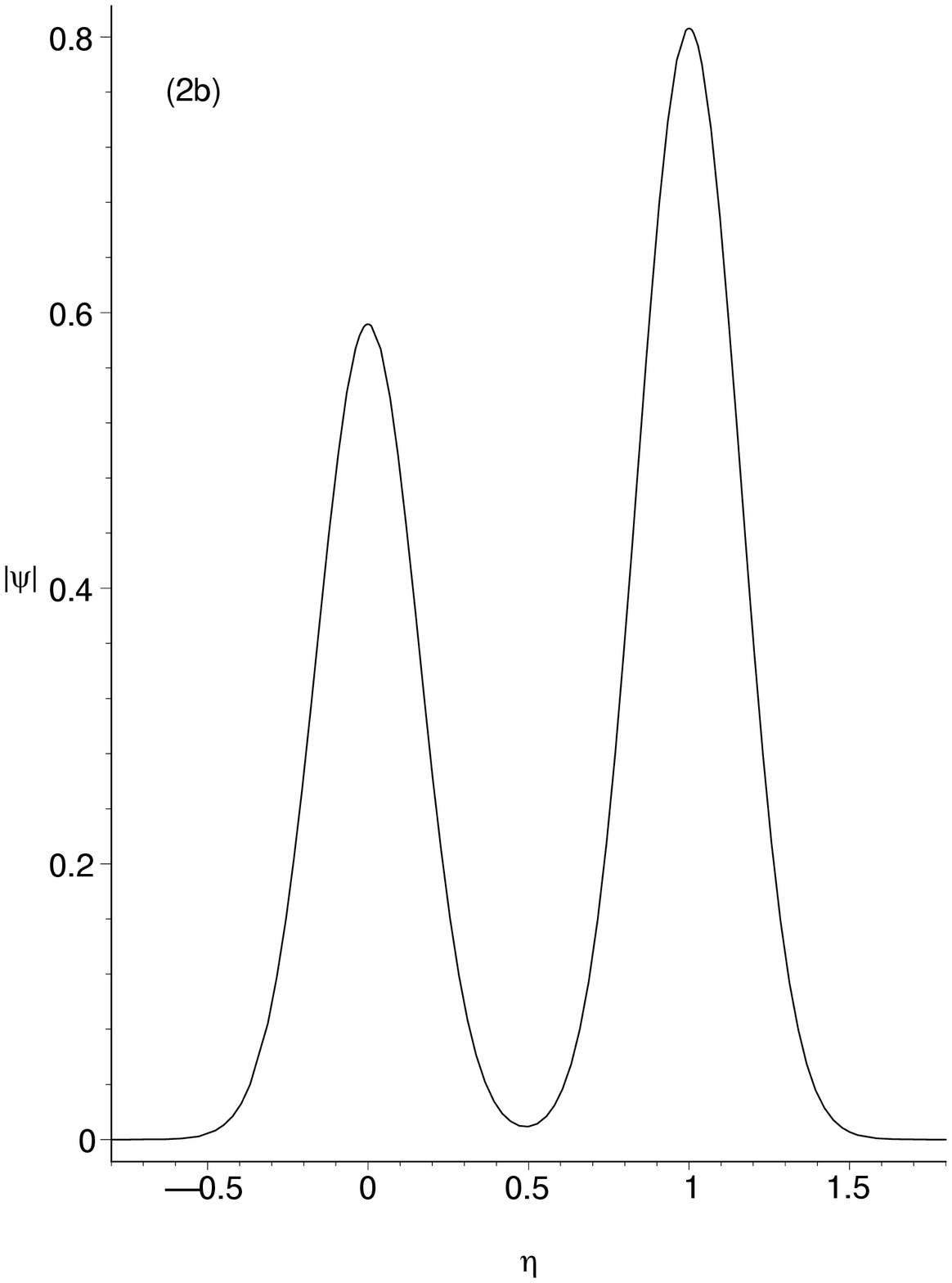}
\includegraphics[height=4cm,width=4cm]{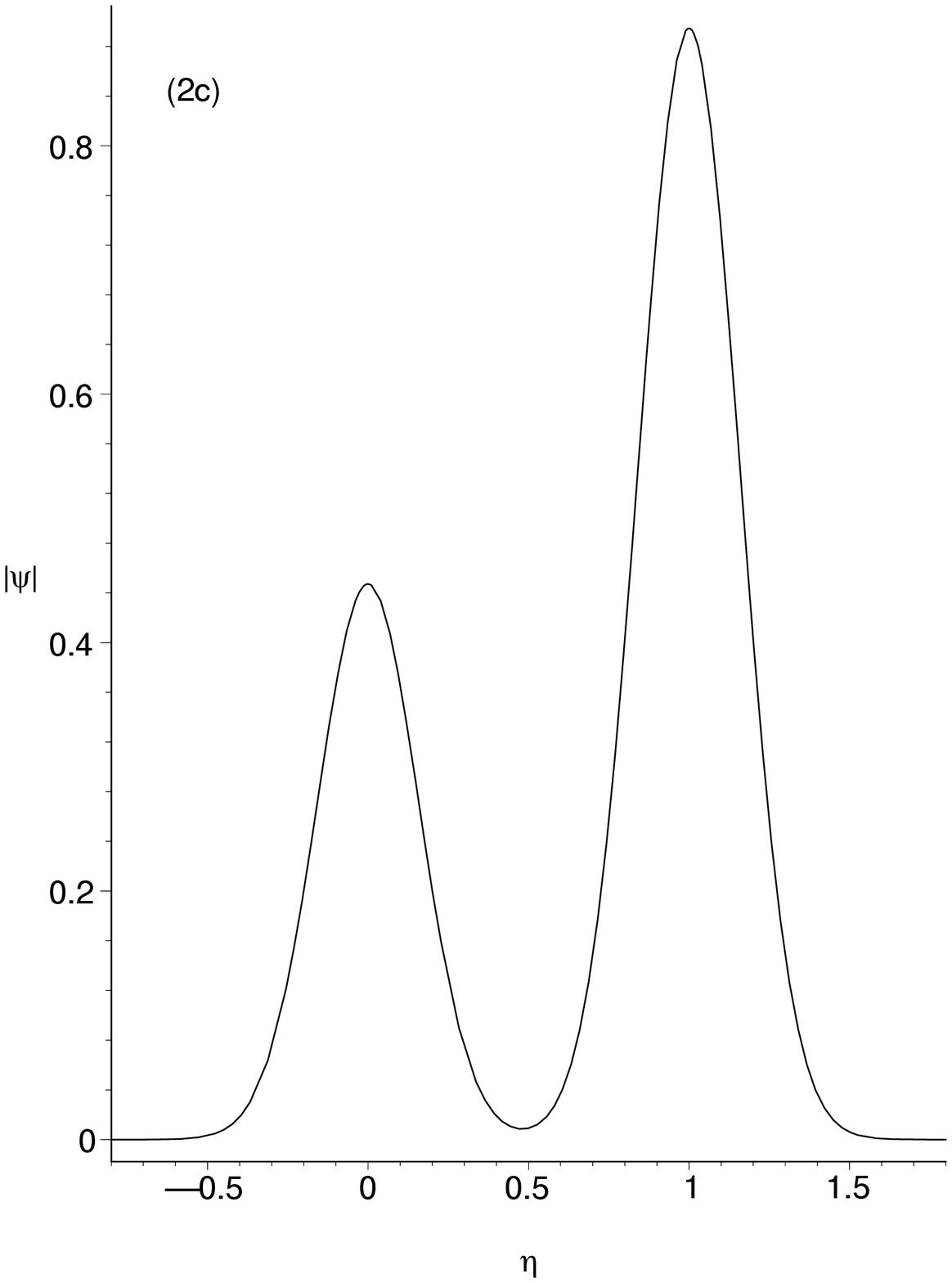}
\caption{Picture of the absolute value of the asymmetrical stationary wave-function $\psi = \sqrt {q_1} \varphi_1 + \sqrt {q_2} \varphi_2 $ for different values of $\eta$ ($\eta =-2.01$ in panel (2a), $\eta =-2.1$ in panel (2b) and $\eta =-2.5$ in panel (2c)) where $\varphi_j$, $j=1,2$, are the vectors, defined in (\ref {F9}), localized on the single lattice cell $\x_j$. \ The asymmetrical solution is, for $\eta $ close to the bifurcation point at $\eta=-2$, delocalized between both two wells; when $|\eta|>2$ increases then the wavefunction is going to be fully localized on just one of the two wells.}
\label {Fig1Bis}
\end{center}
\end{figure}

\subsection {Two-dimensional model} \label {M2} In such a case (\ref {F19Bis}) takes the form
\be
\left \{
\begin {array}{lcl}
\sqrt {q_1 q_2} \sin (\theta_2 - \theta_1 ) + \sqrt {q_1 q_3} \sin (\theta_3 - \theta_1 ) &=& 0 \\ 
\sqrt {q_2 q_1} \sin (\theta_1 - \theta_2 ) + \sqrt {q_2 q_4} \sin (\theta_4 - \theta_2 ) &=& 0 \\ 
\sqrt {q_3 q_1} \sin (\theta_1 - \theta_3 ) + \sqrt {q_3 q_4} \sin (\theta_4 - \theta_3 ) &=& 0 \\ 
\sqrt {q_4 q_2} \sin (\theta_2 - \theta_4 ) + \sqrt {q_4 q_3} \sin (\theta_3 - \theta_4 ) &=& 0 \\ 
- \left ( \sqrt {q_2/q_1} \cos (\theta_2 - \theta_1 ) +\sqrt {q_3/q_1}\cos (\theta_3 - \theta_1 ) \right ) + \eta q_1^\sigma &=& \Omega \\ 
- \left ( \sqrt {q_1/q_2} \cos (\theta_1 - \theta_2 ) +\sqrt {q_4/q_2}\cos (\theta_4 - \theta_2 ) \right ) + \eta q_2^\sigma &=& \Omega \\ 
- \left ( \sqrt {q_1/q_3} \cos (\theta_1 - \theta_3 ) +\sqrt {q_4/q_3}\cos (\theta_4 - \theta_3 ) \right ) + \eta q_3^\sigma &=& \Omega \\ 
- \left ( \sqrt {q_2/q_4} \cos (\theta_2 - \theta_4 ) +\sqrt {q_3/q_4}\cos (\theta_3 - \theta_4 ) \right ) + \eta q_4^\sigma &=& \Omega \\ 
q_1+q_2+q_3+q_4 &=& 1
\end {array}
\right. \label {F27Bis}
\ee
First of all we remark that for dimension $d>1$ equations (\ref {F27Bis}) admit solutions with $q_j=0$ for some $j=0$; in such a case it turns out that the solutions of equations (\ref {F19Bis}), or (\ref {F27Bis}), with some $d_j=0$ are necessarily of the form $d_2=d_3=0$ and $d_1 =-d_4$, or $d_1=d_4$ and $d_2 =-d_3$. \ Then, these solutions are the continuation of the solutions $v_2$ and $v_3$ given in (\ref {F21Bis}) at $\eta =0$. \ On the other side, we remark that at $\eta =0$ then the above equation has ground state corresponding to the solution $v_1$ given in (\ref {F21Bis}) where $q_j = \frac 14$ and $\theta_j =\theta_\ell$ for any value of the indexes $j,\ell =1,2,3,4$. \ By continuity, for $\eta \not= 0$ then the continuation of the solution $v_1$ will have all $q_j >0$ and $\theta_j = \theta_\ell$, for any $ j,\ell =1,2,3,4$. \ In fact, if we assume that, for instance, $\theta_2 > \theta_1$, then the first equation of (\ref {F27Bis}) implies that $\theta_3 - \theta_1 <0$ and the second one implies that $\theta_4 - \theta_2 <0$; hence the third equation implies that $\theta_4 - \theta_3 >0$, and, finally, we have a contradiction because of the fourth equation.

In conclusion, in order to find the bifurcations from the ground state solution we can restrict ourselves to study the following system of equations
\be
f_j =0 , \ j=1,\ldots , 5, \ \mbox { where } \ 
\left \{
\begin {array}{lcl}
f_1 &=& - \left ( \sqrt {q_2/q_1}  +\sqrt {q_3/q_1} \right ) + \eta q_1^\sigma - \Omega \\ 
f_2 &=& - \left ( \sqrt {q_1/q_2}  +\sqrt {q_4/q_2} \right ) + \eta q_2^\sigma - \Omega \\ 
f_3 &=& - \left ( \sqrt {q_1/q_3}  +\sqrt {q_4/q_3} \right ) + \eta q_3^\sigma - \Omega \\ 
f_4 &=& - \left ( \sqrt {q_2/q_4}  +\sqrt {q_3/q_4} \right ) + \eta q_4^\sigma - \Omega \\ 
f_5 &=& q_1+q_2+q_3+q_4 - 1
\end {array}
\right. \label {F21}
\ee
which always has a symmetric solution
\be
q_1=q_2=q_3=q_4 = \frac 14 \ \mbox { with } \ \Omega = \eta 4^{-\sigma} -2 
\label {F22}
\ee

Let us remark that the problem is invariant under the $8$ transformation of the Dihedral group of the square. \ Hence, asymmetrical solutions, if there exist, are degenerate.

Then we look for the symmetric and partially symmetric solutions coming from the solution (\ref {F22}) by bifurcation. 

\subsubsection {Symmetric solutions - mirror symmetry} \label {MS} We look for the solutions such that 
\bee
q_1 = q_2 \ \mbox { and } \ q_3 =q_4 
\eee
and similarly such that $q_1 =q_3$ and $q_2 =q_4$. \ Under these conditions then  equations (\ref {F21}) become
\be
\left \{
\begin {array}{lcl}
- \left ( 1  +\sqrt {q_3/q_1} \right ) + \eta q_1^\sigma &=& \Omega \\ 
- \left ( 1+ \sqrt {q_1/q_3}  \right ) + \eta q_3^\sigma &=& \Omega \\ 
q_1+q_3 &=& \frac 12
\end {array}
\right. \label {F23}
\ee
If we set $z=2(q_1-q_3) \in (-2,+2)$, that is 
\bee
q_1 = \frac 14 (1+ z) \ \mbox { and } \ q_3 = \frac 14 (1-z) 
\eee
then equation (\ref {F23}) takes the form (\ref {F26Bis}), provided we set $\chi =  \frac {\eta}{4^\sigma}$, which has solutions

\begin {itemize}

\item [-] $z=0$, which corresponds to the case (\ref {F22});

\item [-] for $\sigma=1$, $ z = \pm \frac {1}{\eta} \sqrt {\eta^2 - 16}$ for $\eta \le -4$ which gives $\Omega = - 1 + \frac 12 \eta$.

\end {itemize}

Thus, at $\eta =-4$ the solution (\ref {F22}) bifurcates.

\subsubsection {Symmetric solutions - point symmetry} \label {PS} We look for solutions such that 
\be
q_1 = q_4 \ \mbox { and } \ q_2 =q_3 \label {F230Bis}
\ee
Then equations (\ref {F21}) become
\be
\left \{
\begin {array}{lcl}
- 2 \sqrt {q_2/q_1} + \eta q_1^\sigma &=& \Omega \\ 
- 2 \sqrt {q_1/q_2} + \eta q_2^\sigma &=& \Omega \\ 
q_1+q_2 &=& \frac 12
\end {array}
\right. \label {F24}
\ee
If, similarly to the previous case,  we set  $z=2(q_1-q_2) \in (-2,+2)$ then such an equation takes the form(\ref {F26Bis}) provided we set $\chi =  \frac {\eta}{2\cdot 4^\sigma}$, which has solutions

\begin {itemize}

\item [-] $z=0$, which corresponds to the case (\ref {F22});

\item [-] for $\sigma=1$, $ z = \pm \frac {1}{\eta} \sqrt {\eta^2 - 64}$ for $\eta \le -8$ which gives $\Omega =  + \frac 12 \eta$.

\end {itemize}

Thus, at $\eta =-8$ the solution (\ref {F22}) bifurcates again.

\subsubsection {Partially symmetric solutions} We consider now solutions such that 
\be
q_1 = q_4 \ \mbox { and } \ q_2 \not= q_3 \, , \label {F23Ter}
\ee
or similarly such that $q_2 =q_3$ and $q_1 \not= q_4$. In such a case the numerical analysis of equation (\ref {F21}) shows that at $\eta^2_{crit} =-3.5836$ a saddle point occurs and the solution has two branches. \ One branch denoted as branch (a) is connected with the branch of solutions (\ref {F22}) at $\eta =-4$, while on the other branch, denoted as branch (b), $\Omega$ behaves like $\eta$ (see Fig. \ref {Fig2}). \ The relevant fact is that on the branch (b) the wavefunction is going to be well localized on only one well. \ For instance, at $\eta =\eta^2_{crit}$ the value of $q_2$ is equal to $0.571$, and at $\eta =-3.95$ the value of $q_2$ of the solution on the branch (b) is equal to $0.781$, which means that the wavefunction is practically fully localized on the well around $\x_{(0,1)}$ (see Table \ref {tabella2} for different values of $\eta$, see also Fig. \ref {Fig3} and \ref {Fig3}). 

\subsubsection {Classification of the bifurcations}

Bifurcation points are the solutions $(q_1,q_2,q_3,q_4, \Omega ,\eta )$ of the system of equations (\ref {F21}) under the condition 
\be
\mbox {det} \left ( \frac {\partial f_i}{\partial q_h} \right )_{i,h=1,2,3,4,5} =0 \, , \label {Bif}
\ee
where we denote $q_5 =\Omega$. \ Bifurcations of the symmetric stationary solution (\ref {F22}) are the solutions of equation (\ref {Bif}) with $q_j=\frac 14$, $j=1,2,3,4$; this equation has $2$ solutions $\eta =-2$, with double multiplicity, and $\eta =-4$ with multiplicity one. \ Furthermore, we can numerically compute the other solutions of the system $f_j=0$, $j=1,2,3,4,5$; in particular, we can observe the occurrence of a saddle point and a bifurcation. \ That is: 

\begin {itemize}

\item [1.] At $\eta = \eta^2_{crit} =-3.5836$ a saddle point occurs and the new stationary solutions have asymmetrical wavefunctions such that $q_1 = q_4$ and $q_2 \not= q_3 $ and where the wavefunction corresponding to the branch (b) is localized on one well.

\item [2.] At $\eta =-4$ the solution (\ref {F22}) bifurcates. \ We have three branches, one is the branch of the solution (\ref {F22}), another one is the branch (a) associated to the saddle point at $\eta = \eta^2_{crit}$, and the last branch is the branch of the solutions observed in \S \ref {MS}.

\item [3.] At $\eta =-8$ the solution (\ref {F22}) bifurcates again. \ We have two branches, one is the branch of the solution (\ref {F22}) and the other branch is the branch of the solutions observed in \S \ref {PS}.

\item [4.] At $\eta =-8.4853$ the solution (\ref {F230Bis}) bifurcates. \ We have two branches, one is the branch of the solution (\ref {F230Bis}) and the other branch is a branch of solutions of the kind (\ref {F23Ter}).

\end {itemize}

In conclusion, for $\eta >\eta^2_{crit}$ there exists only one solution and it is equally distributed on the four wells; once $\eta$ reaches the value $\eta^2_{crit}$ then two (families of) new solutions suddenly appears, and the solutions associated to the branch (a) of Fig. \ref {Fig2} are fully localized on a single well. \ This phenomenon is the opposite of the one observed in the one-dimensional model where the localization effect gradually occurs, in this case the localization effect suddenly occurs.

\begin{table}
\begin{center}
\begin{tabular}{|l|c|c|c|c|c|c|c|c|c|c|} \hline
\multicolumn{1}{|c|}{} &  \multicolumn{5}{|c|}{Branch (a)} &  \multicolumn{5}{|c|}{Branch (b)}\\ \hline
$\eta $ 			& $q_1$ & $q_2$ & $q_3$ & $q_4$ & $\Omega$ & $q_1$ & $q_2$ & $q_3$ & $q_4$ & $\Omega$  \\ \hline  
$-3.5836$ &  &  & &  &  & $0.168$ & $0.581$ & $0.082$ & $0.168$ & $   -3.159   $\\ \hline 
$-3.6$ & $0.186$ & $0.527$ & $0.1$ & $0.186$ & $   -3.0867   $ & $0.152$ & $0.628$ & $0.067$ & $0.152$ & $   -3.24617   $\\ \hline 
$-3.7$ & $0.214$ & $0.436$ & $0.136$ & $0.214$ & $   -3.152   $ & $0.125$ & $0.703$ & $0.047$ & $0.125$ & $   -3.445   $\\ \hline 
$-3.8$ & $0.229$ & $0.379$ & $0.162$ & $0.229$ & $   -2.998   $& $0.110$ & $0.742$ & $0.037$ & $0.110$ & $   -3.591   $\\ \hline 
$-3.9$ & $0.241$ & $0.329$ & $0.189$ & $0.241$ & $   -2.995   $ & $0.1$ & $0.77$ & $0.031$ & $0.1$ & $   -3.722  $\\ \hline 
$-3.95$ & $0.246$ & $0.302$ & $0.207$ & $0.246$ & $   -2.997   $ & $0.095$ & $0.781$ & $0.029$ & $0.095$ & $   -3.784   $\\ \hline 
$-4.05$ &  &  & &  &  & $0.088$ & $0.801$ & $0.024$ & $0.088$ & $   -3.904   $\\ \hline 
$-8.1$   &  &  & &  &  & $0.016$ & $0.967$ & $0.001$ & $0.016$ & $   -8.088   $\\ \hline 
\end{tabular}
\caption{Here we report the numerical solutions of equations $f_j=0$, $j=1,2,3,4,5$, associated to the two branches (a) and (b) raising from the saddle point at $\eta=\eta^2_{crit} = - 3.5836$. \ We can see that the solutions associated to the branch (b) are going to be fully localized on just one well (in this case it is the well with center at $\x_{(0,1)}$.}
\label{tabella2}
\end{center}
\end {table}

\begin{figure}
\begin{center}
\includegraphics[height=6cm,width=6cm]{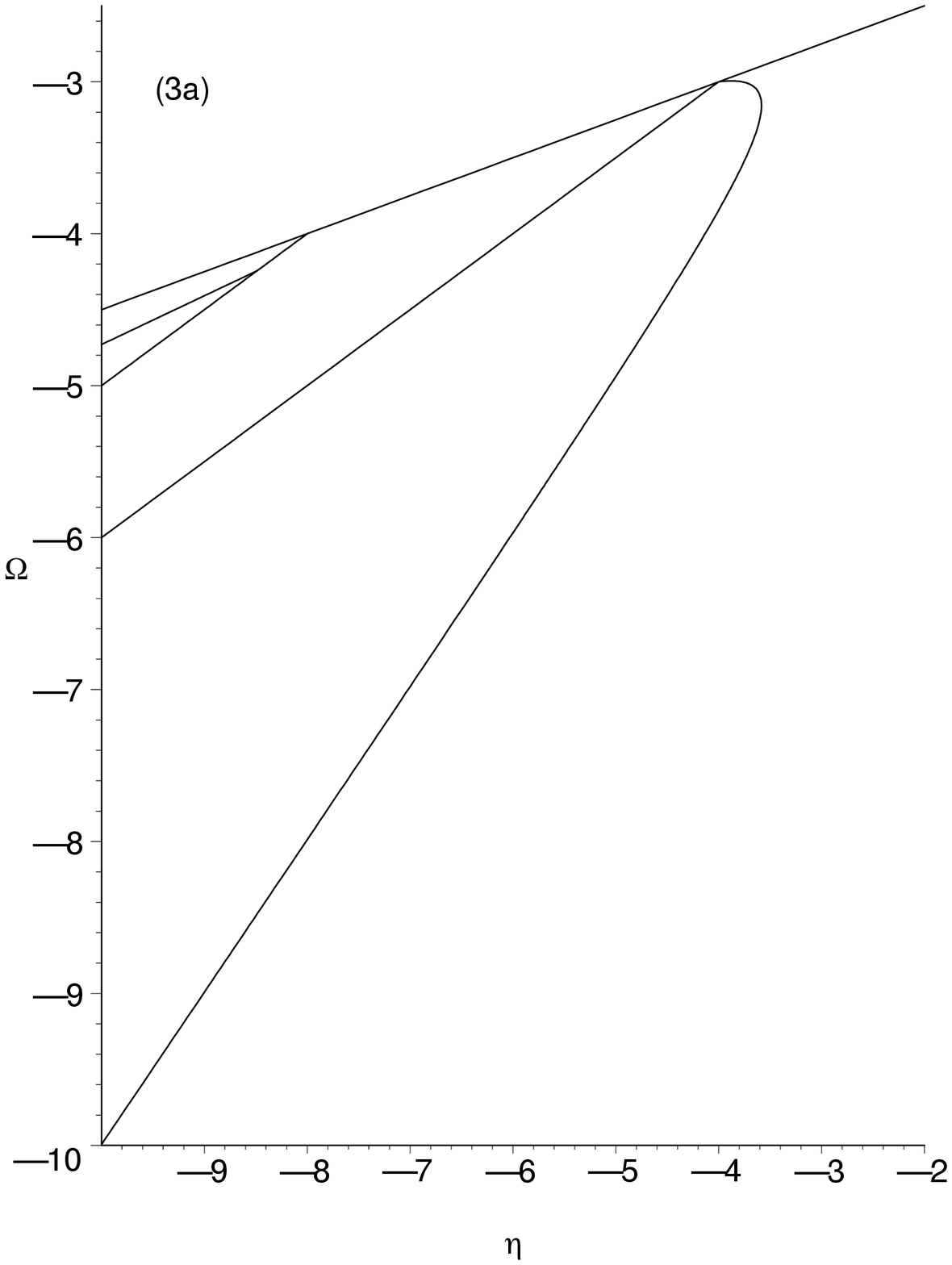}
\includegraphics[height=6cm,width=6cm]{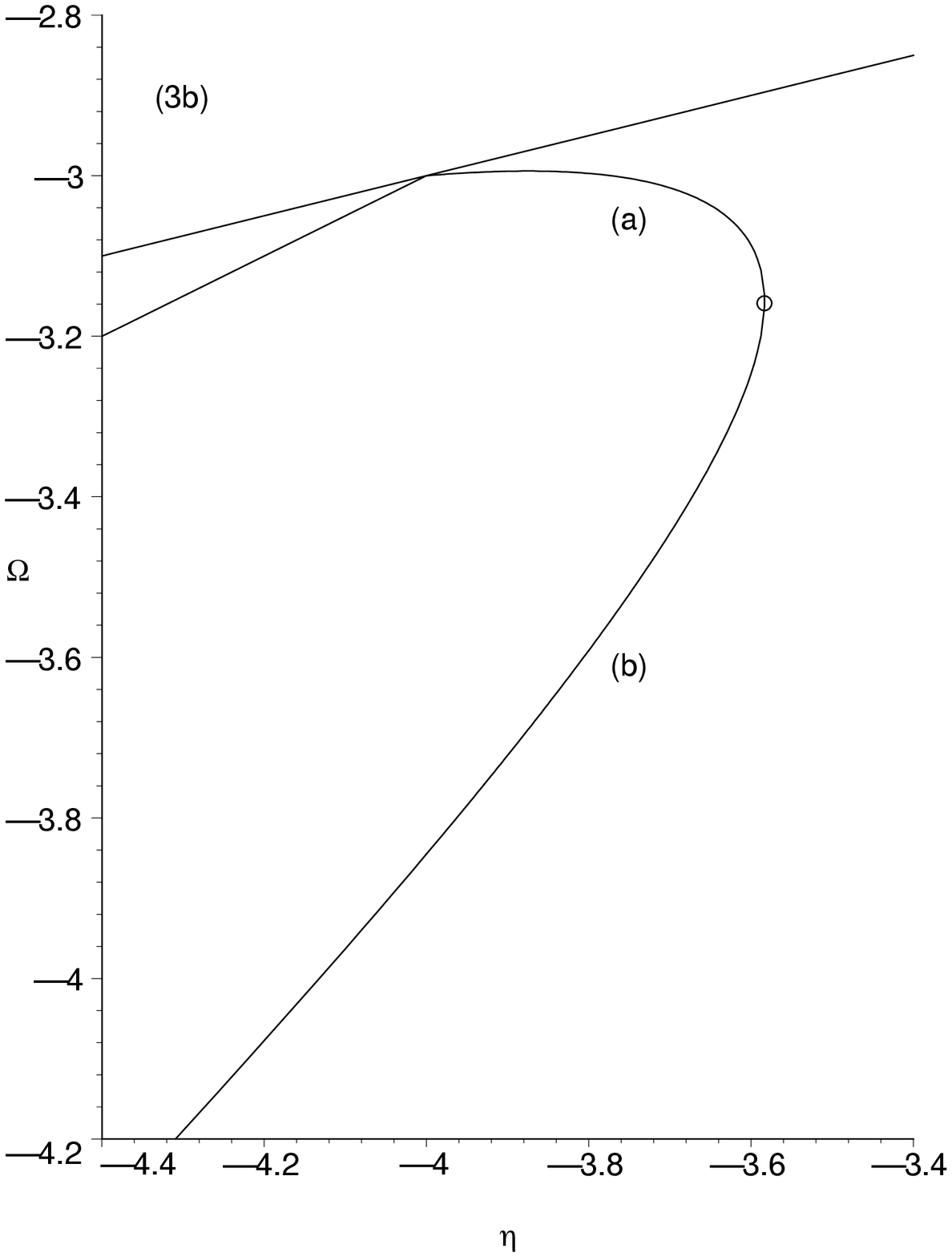}
\caption{In the left panel (3a) we plot the graph of $\Omega$ as function of $\eta$, and we observe $3$ bifurcation points and $1$ saddle point. \ In the right panel (3b) we zoom the picture around the saddle point (denoted by the circle point) and we name the two branches (a) and (b). \ On the branch (b) the wavefunction is going to be fully localized on just one well as $|\eta|$ increases.}
\label {Fig2}
\end{center}
\end{figure}

\begin{figure}
\begin{center}
\includegraphics[height=6cm,width=6cm]{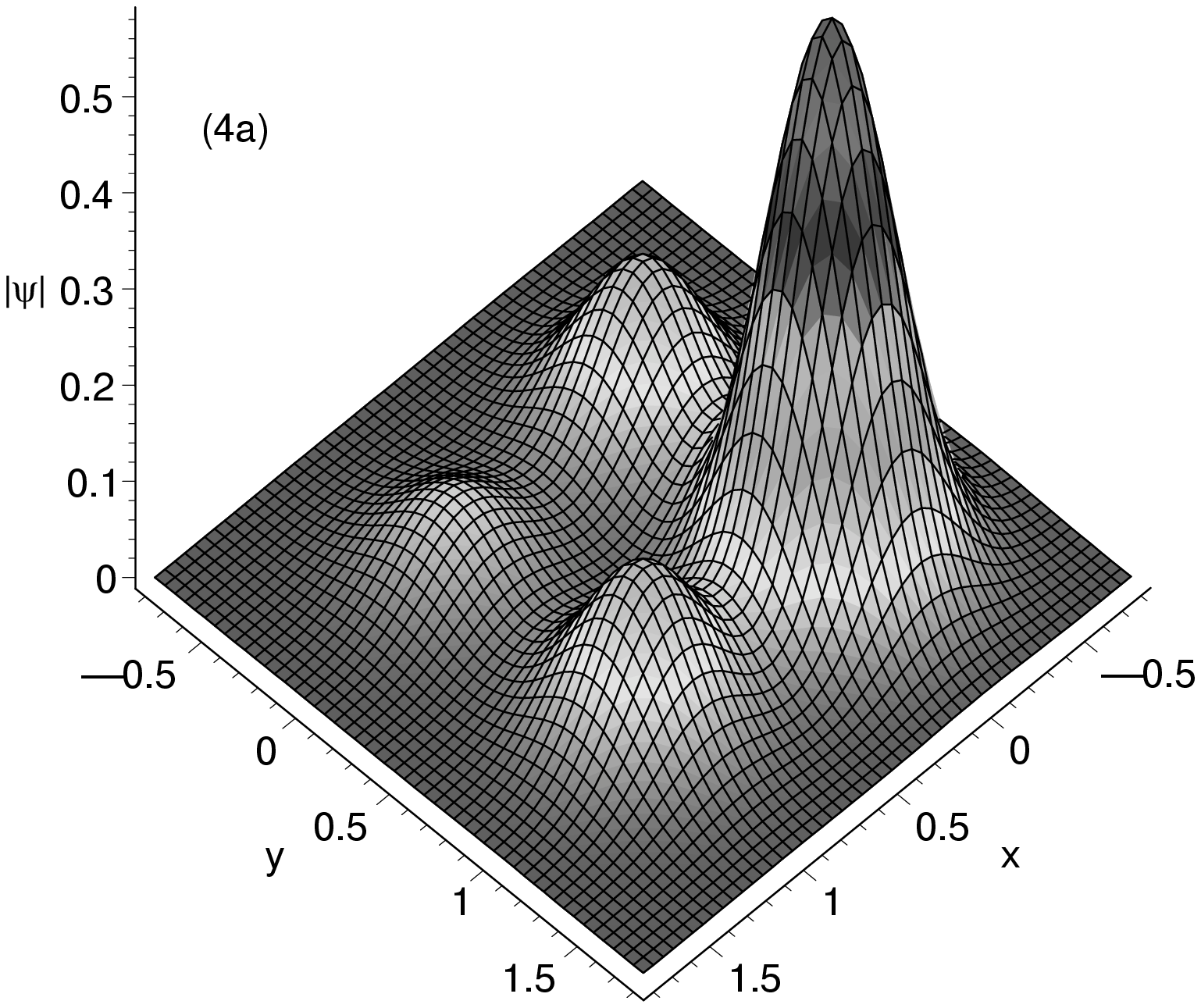}
\includegraphics[height=6cm,width=6cm]{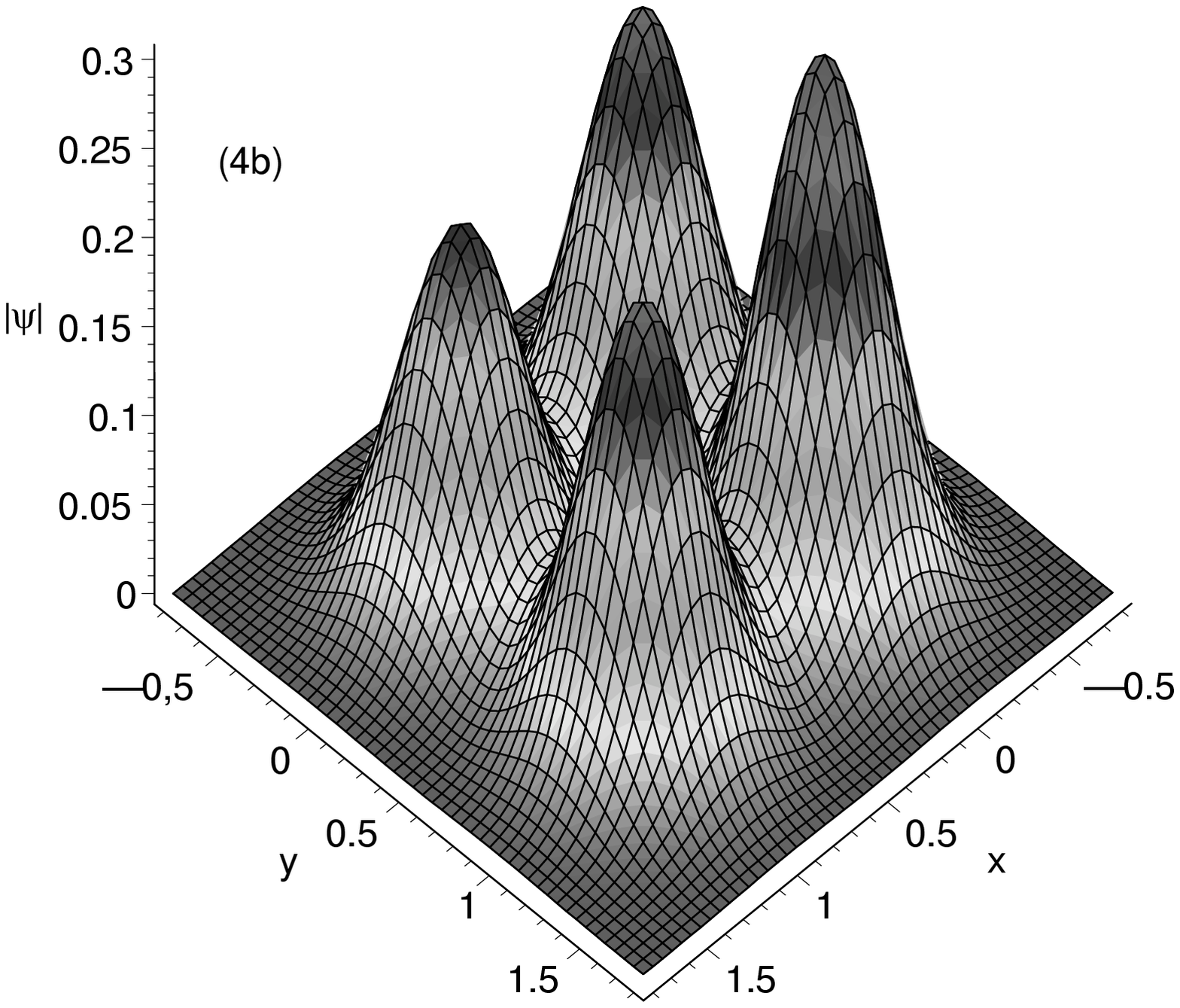}
\includegraphics[height=6cm,width=6cm]{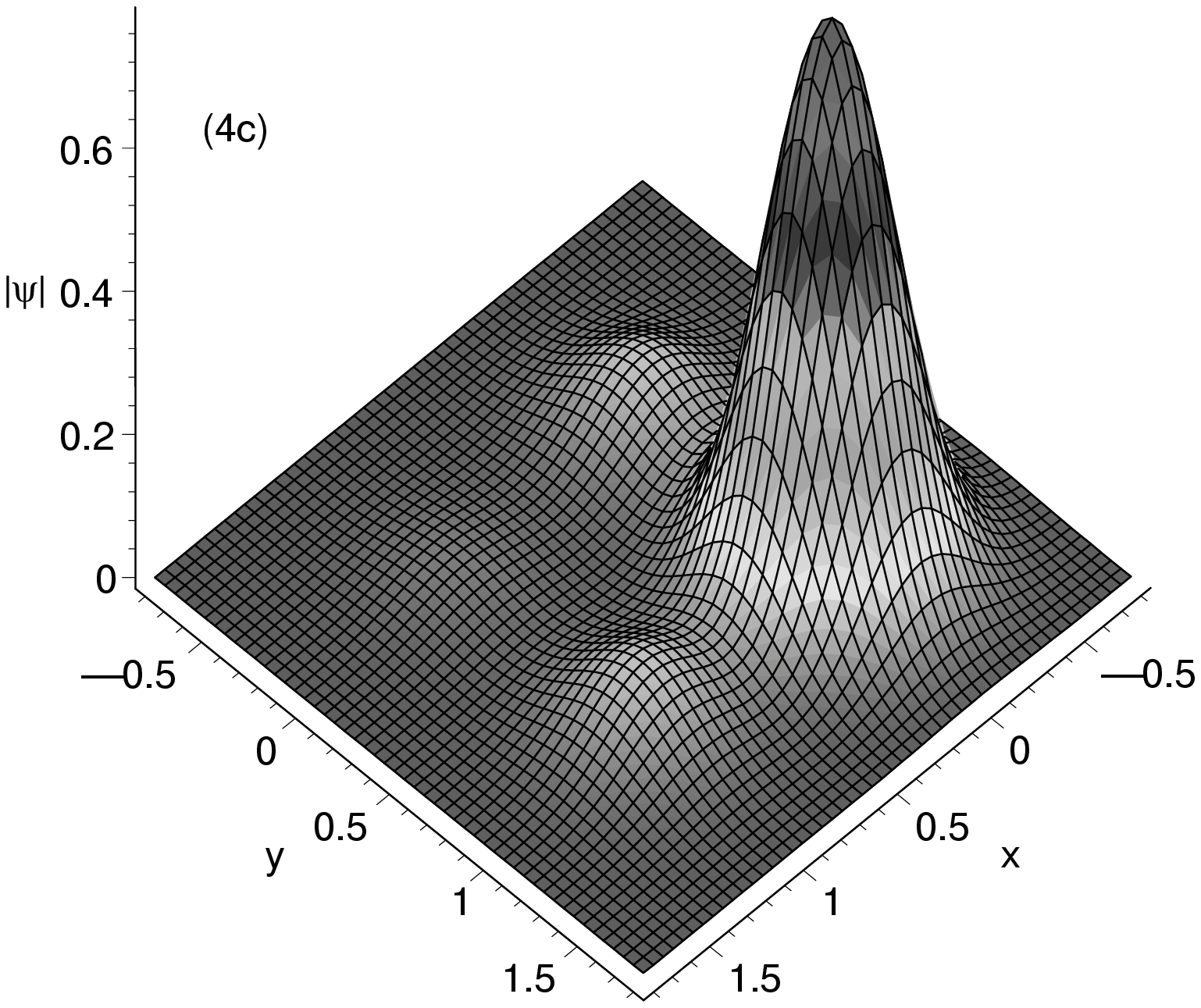}
\caption{Pictures of the absolute value of the asymmetrical stationary wavefunction associated to the branch (b) for $\eta =\eta_{crit}^2$ (panel (4a)) and $\eta =-3.95$; in panel (4b) we plot the absolute value of the function corresponding to a solution in the branch (a); in panel (4c) we plot the absolute value of the function corresponding to a solution in the branch (b). \ At $\eta = \eta_{crit}^2$ the stationary solution is quite localized on just one well and the wavefunction corresponding to the branch (b) is going to be localized on the well around $\x_{(0,1)}$ when $|\eta |$ increases; on the other side the wavefunction corresponding to the branch (a) is going to be delocalized between the four wells as $\eta$ approaches $\eta =-4$.}
\label {Fig3}
\end{center}
\end{figure}

\subsection {Three-dimensional model}

 By means of arguments similar to the discussed in \S \ref {M2} for two-dimensional models it turns out that bifurcations of the continuation of the symmetric solution $v_1$ in (\ref {F23Bis}) of (\ref {F19Bis}) are such that $q_j >0$ and $\theta_j =\theta_\ell$ for any $j,\ell =1,2,\ldots , 8$. \ That is, in order to find the bifurcation from the ground state solution we can restrict ourselves to study the following system of equations
 \be
\left \{ 
\begin {array}{lcl}
- \frac {1}{\sqrt{q_1}} \left ( \sqrt {q_2} + \sqrt {q_3} + \sqrt {q_5} \right ) + \eta q_1^\sigma &=& \Omega \\
- \frac {1}{\sqrt{q_2}} \left ( \sqrt {q_1} + \sqrt {q_4} + \sqrt {q_6} \right ) + \eta q_2^\sigma &=& \Omega \\
- \frac {1}{\sqrt{q_3}} \left ( \sqrt {q_1} + \sqrt {q_4} + \sqrt {q_7} \right ) + \eta q_3^\sigma &=& \Omega \\
- \frac {1}{\sqrt{q_4}} \left ( \sqrt {q_2} + \sqrt {q_3} + \sqrt {q_8} \right ) + \eta q_4^\sigma &=& \Omega \\
- \frac {1}{\sqrt{q_5}} \left ( \sqrt {q_1} + \sqrt {q_6} + \sqrt {q_7} \right ) + \eta q_5^\sigma &=& \Omega \\
- \frac {1}{\sqrt{q_6}} \left ( \sqrt {q_2} + \sqrt {q_5} + \sqrt {q_8} \right ) + \eta q_6^\sigma &=& \Omega \\
- \frac {1}{\sqrt{q_7}} \left ( \sqrt {q_3} + \sqrt {q_5} + \sqrt {q_8} \right ) + \eta q_7^\sigma &=& \Omega \\
- \frac {1}{\sqrt{q_8}} \left ( \sqrt {q_4} + \sqrt {q_6} + \sqrt {q_7} \right ) + \eta q_8^\sigma &=& \Omega \\
\end {array}
\right.
\label {F25}
\ee
which always has a symmetric solution
\be
q_j = \frac 18 , \ j=1, \ldots , 8 \ \mbox { with } \ \Omega = \eta 8^{-\sigma} -3 \, . \label {F26}
\ee 
As in the previous cases in dimension 1 and 2, asymmetrical solutions, if there exists, ar degenerate because of the invariance of the model with respect to several transformations; in particular our model is invariant with respect to the 48 transformations of the achiral octahedral symmetric group isomorphic to $S_4 \times C_2$. \ Among the symmetric and partially symmetric solutions coming from the solution (\ref {F26}) by bifurcation the first one we can observe are the two families of partially symmetric solutions such that
\be
q_2 = q_3 = q_5 \ \mbox { and } \ q_4 =q_6 =q_7 \label {F27}
\ee
and such that 
\be
q_1 =q_2 = q_7 = q_8 \, , \ q_5 =q_6 \ \mbox { and } \ q_3 = q_4 \, . \label {F28}
\ee

In particular (see Figure \ref {Fig4}):

\begin {enumerate}

\item [1.] At $\eta = \eta^3_{crit,1} =  -5.0116$ a saddle point occurs and the new stationary solutions have asymmetrical wavefunctions satisfying (\ref {F27}), and where the wavefunction corresponding to the branch (b) of Fig. \ref {Fig4} is localized on just one single well (see table \ref {tabella3}).

\item [2.] At $\eta = \eta^3_{crit,2} =  -7.1672$ a saddle point occurs and the new stationary solutions have asymmetrical wavefunctions satisfying (\ref {F28}), and where the wavefunction corresponding to the branch (d) of Fig. \ref {Fig4} is localized on a couple of adjacent  wells (see table \ref {tabella4}).

\item [3.] At $\eta =-8$ the solution (\ref {F26}) bifurcates in four branches; two of them are the branches (a) and (c) connected with the two saddle points previously discussed.

\item [4.] Solution (\ref {F26}) bifurcates at the value $\eta =-16$ and $\eta =-24$, too. \ The bifurcation point at $\eta =-16$ corresponds to 4 different branches, the bifurcation point at $\eta =-24$ corresponds to 2 different branches.

\end {enumerate}

\begin{figure}
\begin{center}
\includegraphics[height=6cm,width=6cm]{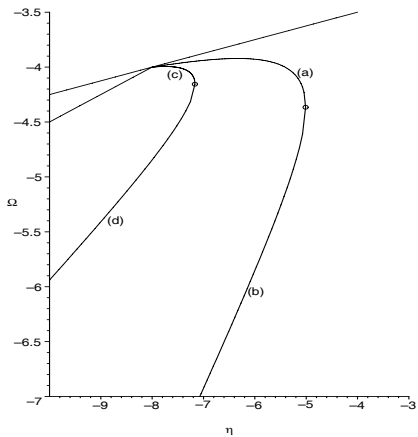}
\caption{We plot the graph of $\Omega$ as function of $\eta$, and we observe $2$ bifurcation points and $2$ saddle point. \ On the branches (b) the wavefunction is going to be fully localized on just one well as $|\eta|$ increases.}
\label {Fig4}
\end{center}
\end{figure}

\begin{table}
\begin{center}
\begin{tabular}{|l|c|c|c|c|c|c|c|c|c|c|} \hline
\multicolumn{1}{|c|}{} &  \multicolumn{5}{|c|}{Branch (a)} &  \multicolumn{5}{|c|}{Branch (b)}\\ \hline
$\eta $ 			& $q_1$ & $q_2$ & $q_4$ & $q_8$ & $\Omega$ & $q_1$ & $q_2$ & $q_4$ & $q_8$ & $\Omega$  \\ \hline  
$-5.0116$ &   &  & &  &  & $0.659$ & $0.083$ & $0.027$ & $0.0129$ & $   -4.3659   $\\ \hline 
$-5.1$ & $0.569$ & $0.098$ & $0.038$ & $0.021$ & $   -4.151   $ & $0.739$ & $0.067$ & $0.017$ & $0.007$ & $   -4.674   $\\ \hline 
$-6$ & $0.355$ & $0.128$ & $0.071$ & $0.048$ & $   -3.928   $ & $0.873$ & $0.037$ & $0.005$ & $0.001$ & $   -5.853   $\\ \hline 
$-7$ & $0.240$ & $0.136$ & $0.094$ & $0.071$ & $   -3.939   $& $0.919$ & $0.024$ & $0.002$ & $0.0004$ & $   -6.924   $\\ \hline 
$-7.9$ & $0.150$ & $0.131$ & $0.117$ & $0.105$ & $   -3.993   $ & $0.941$ & $0.018$ & $0.013$ & $0.0002$ & $   -7.852  $\\ \hline 
$-8.1$ &  & & &  &  & $0.945$ & $0.017$ & $0.001$ & $0.0002$ & $   -8.056   $\\ \hline 
$-10$ &  &  & &  &  & $0.957$ & $0.014$ & $0.001$ & $0.0001$ & $   -8.970   $\\ \hline 
\end{tabular}
\caption{Here we report the numerical solutions of equations (\ref {F25}) under conditions (\ref {F27}) associated to the two branches (a) and (b) raising from the saddle point at $\eta=\eta^3_{crit,1} = - 5.0116$. \ We can see that the solutions associated to the branch (b) are going to be fully localized on just one well.}
\label{tabella3}
\end{center}
\end {table}

\begin{table}
\begin{center}
\begin{tabular}{|l|c|c|c|c|c|c|c|c|} \hline
\multicolumn{1}{|c|}{} &  \multicolumn{4}{|c|}{Branch (c)} &  \multicolumn{4}{|c|}{Branch (d)}\\ \hline
$\eta $ 			& $q_1$ & $q_3$ & $q_5$ & $\Omega$ & $q_1$ & $q_3$ & $q_5$ & $\Omega$  \\ \hline  
$-7.1672$ &  &  &  &  & $0.085$ & $0.041$ & $0.290$  & $   -4.156   $\\ \hline 
$-7.2$ & $0.093$ & $0.050$  & $0.264$ & $-4.087$ & $0.076$ & $0.034$ & $0.314$  & $   -4.246   $\\ \hline 
$-7.5$ & $0.111$ & $0.074$ & $0.203$  & $-4.003$ & $0.059$ & $0.021$ & $0.362$  & $   -4.521   $\\ \hline 
$-7.9$ & $0.123$ & $0.103$  & $0.150$ & $-3.997$ & $0.048$ & $0.014$ & $0.391$  & $   -4.784   $\\ \hline 
$-8.1$ &  &  &  &  & $0.044$ & $0.012$ & $0.400$  & $   -4.904   $\\ \hline 
$-9$ &  &  &  &  & $0.032$ & $0.007$ & $0.429$  & $   -5.409   $\\ \hline 
\end{tabular}
\caption{Here we report the numerical solutions of equations (\ref {F25}) under conditions (\ref {F28}) associated to the two branches (c) and (d) raising from the saddle point at $\eta=\eta^3_{crit,2} = -7.1672$. \ We can see that the solutions associated to the branch (d) are going to be fully localized on a couple of wells.}
\label{tabella4}
\end{center}
\end {table}

\subsection {Conclusion} As we can see in the the previous pictures and tables, there exists a fundamental difference between the one-dimensional model and the two- and three-dimensional models: the appearance of saddle points associated to branch of stationary solutions localized on a single well. \ This fact is the basic argument for the explanation of the phase transition from superfluidity phase to insulator phase. \ Indeed, in presence of stationary solutions associated to the ground state and localized in just one well we expect that the typical beating motion in symmetric potential does not work and thus the motion of the particle of the condensate between adjacent wells is forbidden. \ Since in dimension 1 the asymmetrical state becomes gradually localized on just one well when the nonlinear strength parameter increases, then the phase transition is quite slow. \ In dimension 2 and 3 we have the opposite situation, the asymmetrical ground states localized on just one well suddenly appear with the saddle points and then the phase transition is expected to be very sharp.

\end {document}